\documentclass[letterpaper, 10 pt, conference]{ieeeconf}
\IEEEoverridecommandlockouts
\usepackage{cite}
\usepackage{amsmath,amssymb,amsfonts}
\usepackage{algorithmic}
\usepackage{graphicx}
\usepackage{textcomp}
\usepackage{xcolor}
\usepackage{mathrsfs}
\usepackage{psfrag}
\begin{document}

\title{\LARGE \bf Robust Output Feedback Consensus for Networked Heterogeneous Nonlinear Negative-Imaginary Systems\\
}

\author{Kanghong Shi$^{\dagger}$,\qquad Igor G. Vladimirov$^{\dagger}$,\qquad Ian R. Petersen$^{\dagger}$
\thanks{This work was supported by the Australian Research Council under grant DP190102158.}
\thanks{$^{\dagger}$Research School of Electrical, Energy and Materials Engineering, College of Engineering and Computer Science, Australian National University, Canberra, Acton, ACT 2601, Australia.
        {\tt kanghong.shi@anu.edu.au}, {\tt igor.vladimirov@anu.edu.au}, {\tt ian.petersen@anu.edu.au}.}%
}

\newtheorem{definition}{Definition}
\newtheorem{theorem}{Theorem}
\newtheorem{conjecture}{Conjecture}
\newtheorem{lemma}{Lemma}
\newtheorem{remark}{Remark}
\newtheorem{corollary}{Corollary}
\newtheorem{assumption}{Assumption}

\maketitle
\thispagestyle{empty}
\pagestyle{empty}

\begin{abstract}
This paper provides a control protocol for the robust output feedback consensus of networked heterogeneous nonlinear negative-imaginary (NI) systems. Heterogeneous nonlinear output strictly negative-imaginary (OSNI) controllers are applied in positive feedback according to the network topology to achieve output feedback consensus. The main contribution of this paper is extending the previous studies of the robust output feedback consensus problem for networked heterogeneous linear NI systems to nonlinear NI systems. Output feedback consensus is proved by investigating the internal stability of the closed-loop interconnection of the network of heterogeneous nonlinear NI plants and the network of heterogeneous nonlinear OSNI controllers according to the network topology. The network of heterogeneous nonlinear NI systems is proved to be also a nonlinear NI system, and the network of heterogeneous nonlinear OSNI systems is proved to be also a nonlinear OSNI system. Under suitable conditions, the nonlinear OSNI controllers lead to the convergence of the outputs of all nonlinear NI plants to a common limit trajectory, regardless of the system model of each plant. Hence, the protocol is robust with respect to parameter perturbation in the system models of the heterogeneous nonlinear NI plants in the network.
\end{abstract}


\section{Introduction}
Negative-imaginary (NI) systems theory was introduced by Lanzon and Petersen in \cite{lanzon2008stability,petersen2010feedback} in order to explain the robustness properties in a class of controllers for flexible structures \cite{fanson1990positive}. NI systems theory has attracted a lot of interest among control theory researchers (see \cite{xiong2010negative,mabrok2014generalizing,xiong2009lossless,song2012negative,mabrok2011spectral}, etc.). NI systems theory complements positive-real (PR) systems theory because it can be applied to systems with a relative degree from zero to two, while PR systems theory can only deal with systems with relative degree of zero or one. Typical NI systems are mechanical systems with co-located force actuators and position sensors. Positive-position feedback control is often used for NI systems, which can be applied to flexible structures with highly resonant dynamics due to the robustness of NI systems with respect to uncertainty in system models and external disturbances. NI systems theory has already achieved success in some fields, such as nano-positioning (see \cite{das2014resonant,das2014mimo,das2015multivariable}, etc.).

NI systems theory was recently extended to nonlinear systems \cite{ghallab2018extending}. A system is said to be nonlinear NI if it is dissipative with the supply rate $w=u^T\dot y$, where $u$ and $y$ are the input and output of the system, respectively, and $y$ only depends on the system state $x$. While the positive-feedback interconnection of a linear NI system and a linear strictly negative-imaginary (SNI) system is asymptotically stable if their cascaded DC gain is less than unity, the positive-feedback interconnection of a nonlinear NI system and a so-called weak strictly nonlinear NI system is proved in \cite{ghallab2018extending} to be also asymptotically stable under reasonable assumptions.

A class of linear NI systems called output strictly negative-imaginary (OSNI) systems was introduced in \cite{bhowmick2017lti} and \cite{bhowmickoutput} for linear systems and was recently extended to nonlinear systems in \cite{shi2020robust}. A system is said to be nonlinear OSNI if it is dissipative with the supply rate $w(u,\dot y)=u^T\dot y-\epsilon \left|\dot y\right|^2$, where $u$, $x$ and $y$ are the input, state and output of the system, respectively. Also, $y$ is only dependent on $x$. Here, $\epsilon>0$ quantifies the level of output strictness of the system. It is proved in \cite{shi2020robust} that the closed-loop interconnection of a nonlinear NI system and a nonlinear OSNI system is asymptotically stable under certain conditions.

Consensus problems have been widely studied by control theory researchers (see \cite{kim2010output,xi2012output,hu2016output,li2015output}, etc). 
NI systems theory was used in \cite{wang2015robust} to address an output feedback consensus problem for systems with colocated force actuators and position sensors. In \cite{wang2015robust}, a network of systems is said to have output feedback consensus if the outputs of all subsystems converge to a common limit trajectory under the effect of the network communication between subsystems. With certain conditions satisfied, the outputs of heterogeneous linear NI systems connected according to an undirected connected graph can converge to the same limit trajectory if edge-based linear SNI controllers are connected to the plants in positive feedback according to the network topology. The results in \cite{wang2015robust} has been used to address cooperative control problems for multiple UAVs and robots (see \cite{tran2017formation,qi2016cooperative}). However, the result of \cite{wang2015robust} is only applicable for linear NI systems. Motivated by the nonlinear NI systems theory, this paper extends the results in \cite{wang2015robust} to nonlinear NI systems by using the results in \cite{ghallab2018extending} and \cite{shi2020robust}.

The main contribution of this paper is providing a control framework to synchronise multiple heterogeneous nonlinear NI systems under certain conditions. Output feedback consensus of networked heterogeneous nonlinear NI systems is proved by analysing the stability of the closed-loop interconnection of a network of heterogeneous nonlinear NI plants and a network of heterogeneous nonlinear OSNI controllers. Output consensus is guaranteed as long as the nonlinear NI property of the networked systems is preserved and suitable conditions are satisfied. An example is provided to demonstrate the effectiveness of the protocol in dealing with networked heterogeneous nonlinear NI systems.

Notation: The notation in this paper is standard. $\mathbb R$ denotes the field of real numbers. $\mathbb R^{m\times n}$ denotes the space of real matrices of dimension $m\times n$. $A^T$ denotes the transpose of matrix $A$. $\overline{u}$ denotes a constant vector or scalar. $I_n$ is the $n\times n$ identity matrix. $A\otimes B$ denotes the Kronecker product of matrices $A$ and $B$. For a nonlinear system $H$ with input $u$ and output $y$, $y=H(u)$ describes its input-output relationship. 

Graph theory preliminaries: $\mathcal G=(\mathcal V,\mathcal E)$, where $\mathcal V=\{v_1,v_2,\cdots,v_N\}$ and $\mathcal E=\{e_1,e_2,\cdots,e_l\} \subseteq \mathcal V\times \mathcal V$, describes an undirected graph with $n$ nodes and $l$ edges. The symmetric adjacency matrix $\mathcal A = [a_{ij}]\in \mathbb R^{N\times N}$ is defined so that $a_{ii}=0$, and $\forall i\neq j$, $a_{ij}=1$ if $(v_i,v_j)\in \mathcal E$ and $a_{ij}=0$ otherwise. A sequence of unrepeated edges in $\mathcal E$ that joins a sequence of nodes in $\mathcal V$ defines a path. An undirected graph is connected if there is a path between every pair of nodes. Given an undirected graph $\mathcal G$, a corresponding directed graph can be obtained by defining a direction for each edge of $\mathcal G$. The incidence matrix $\mathcal Q=[q_{ev}]\in \mathbb{R}^{l\times N}$ of a directed graph is defined so that the elements in $\mathcal Q$ are given by
\begin{equation*}
    q_{ev}:=
    \begin{cases}
    1 & \text{if } v \text{ is the initial vertex of edge } e,\\
    -1 & \text{if } v \text{ is the terminal vertex of edge } e,\\
    0 & \text{if } v \text{ does not belong to edge } e.
    \end{cases}
\end{equation*}
In this paper, the initial and terminal vertices of an edge in a directed graph can both send information to each other. For an undirected graph $\mathcal G$, the choice of a corresponding directed graph is not unique. However, the Laplacian matrix $\mathcal L_N$ of $\mathcal G$ has the following relationship with the incidence matrix $\mathcal Q$ of any directed graph corresponding to $\mathcal G$: $\mathcal{L}_N=\mathcal{Q}^T\mathcal{Q}$. 
\section{Preliminaries}
Here, we recall the definitions of nonlinear negative-imaginary systems and nonlinear output strictly negative-imaginary systems.

Consider the following general nonlinear system:
\begin{align}
    \dot x(t)=&\ f(x(t),u(t));\label{eq:state equation of nonlinear OSNI}\\
    y(t)=&\ h(x(t))\label{eq:output equation of nonlinear OSNI}
\end{align}
where $f:\mathbb R^n\times \mathbb R^m \to \mathbb R^n$ is a Lipschitz continuous function and $h:\mathbb R^n \to \mathbb R^m$ is a class $C^1$ function.
\begin{definition}\cite{shi2020robust}\label{def:nonlinear NI}
The system (\ref{eq:state equation of nonlinear OSNI}), (\ref{eq:output equation of nonlinear OSNI}) is said to be a nonlinear negative-imaginary system if there exists a positive definite storage function $V:\mathbb R^n\to \mathbb R$ of class $C^1$ such that
\begin{equation}\label{eq:NI MIMO definition inequality}
    \dot V(x(t))\leq u(t)^T\dot y(t),\qquad \forall t\geq 0.
\end{equation}
\end{definition}

\begin{definition}\cite{shi2020robust}\label{def:nonlinear OSNI}
The system (\ref{eq:state equation of nonlinear OSNI}), (\ref{eq:output equation of nonlinear OSNI}) is said to be a nonlinear output strictly negative-imaginary system if there exists a positive definite storage function $V:\mathbb R^n\to\mathbb R$ of class $C^1$ and a constant $\epsilon>0$ such that
\begin{equation}\label{eq:dissipativity of OSNI}
    \dot V(x(t))\leq u(t)^T\dot y(t) -\epsilon \left|\dot y(t)\right|^2, \qquad \forall t\geq 0.
\end{equation}
The quantity $\epsilon$ describes the level of output strictness of the system.
\end{definition}
\section{Robust output feedback consensus}
\begin{figure}[h!]
\centering
\psfrag{H_0}{$\mathcal{H}_p$}
\psfrag{u_p1}{$u_{p1}$}
\psfrag{y_p1}{$y_{p1}$}
\psfrag{u_p2}{$u_{p2}$}
\psfrag{y_p2}{$y_{p2}$}
\psfrag{u_pn}{$u_{pN}$}
\psfrag{y_pn}{$y_{pN}$}
\psfrag{H_p1}{\hspace{-0.05cm}$H_{p1}$}
\psfrag{H_p2}{\hspace{-0.05cm}$H_{p2}$}
\psfrag{H_pn}{\hspace{-0.05cm}$H_{pN}$}
\psfrag{ddd}{\hspace{0.13cm}$\vdots$}
\psfrag{udd}{\hspace{0.15cm}$\vdots$}
\psfrag{odd}{\hspace{0.1cm}$\vdots$}
\hspace{0.5cm}\vspace{-0.3cm}\includegraphics[width=7.5cm]{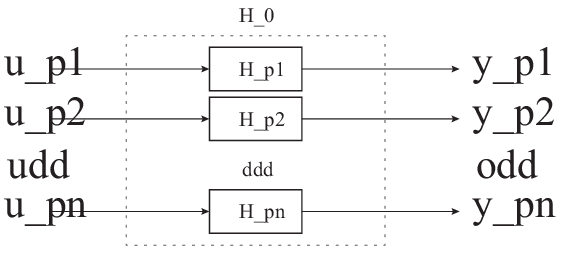}
\caption{System $\mathcal{H}_p$: a nonlinear system consisting of $N$ independent and heterogeneous nonlinear systems $H_{pi}$ $(i=1,2,\cdots,N)$, with independent inputs and outputs combined as the input and output of the networked system $\mathcal{H}_p$.}
\label{fig:H_p1_hetero}
\end{figure}
Consider $N$ heterogeneous nonlinear systems $H_{pi}$ $(i=1,2,\cdots,N)$ described as
\begin{align}
    \dot x_{pi}(t)=&\ f_{pi}(x_{pi}(t),u_{pi}(t));\label{eq:state pi}\\
    y_{pi}(t)=&\ h_{pi}(x_{pi}(t))\label{eq:output pi}
\end{align}
where $f_{pi}:\mathbb R^n\times \mathbb R^m \to \mathbb R^n$ is a Lipschitz continuous function and $h_{pi}:\mathbb R^n \to \mathbb R^m$ is a class $C^1$ function. They operate independently in parallel and each of them has its own input $u_{pi}\in \mathbb{R}^{m}$ and output $y_{pi}\in\mathbb{R}^{m}$, ($i=1,2,\cdots,N$), which is shown in Fig.~\ref{fig:H_p1_hetero}. The subscript ``$p$" indicates that this system will play the role of a plant in what follows. We combine the inputs and outputs respectively as the vectors $U_p=\left[u_{p1}^T, u_{p2}^T, \cdots,  u_{pN}^T\right]^T\in \mathbb R^{Nm\times 1}$, and $Y_p=\left[y_{p1}^T, y_{p2}^T, \cdots, y_{pN}^T\right]^T\in \mathbb R^{Nm\times 1}$.

Let us consider the networked plants connected according to the graph network topology $\hat{\mathcal{H}}_{p}$ as shown in Fig.~\ref{fig:networked_plants}, where $\mathcal{Q}$ is the incidence matrix of a directed graph that represents the communication links between the heterogeneous nonlinear NI plants.
\begin{figure}[h!]
\centering
\psfrag{H_p}{$\mathcal{H}_p$}
\psfrag{H_p1}{$H_{p1}$}
\psfrag{H_p2}{$H_{p2}$}
\psfrag{H_pl}{$H_{pN}$}
\psfrag{U_p}{\hspace{-0.05cm}$U_{p}$}
\psfrag{Y_p}{$Y_{p}$}
\psfrag{ddots}{$\ddots$}
\psfrag{H_np}{$\hat{\mathcal{H}}_{p}$}
\psfrag{U_np}{\hspace{0.2cm}$\hat{U}_{p}$}
\psfrag{Y_np}{\hspace{0.2cm}$\hat{Y}_{p}$}
\psfrag{Q_t}{$\mathcal{Q}\otimes I_m$}
\psfrag{Q}{\hspace{-0.2cm}$\mathcal{Q}^T\otimes I_m$}
\vspace{-0.2cm}\includegraphics[width=8.5cm]{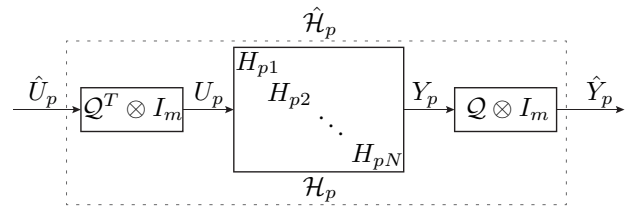}
\caption{Heterogeneous nonlinear NI plants connected according to the directed graph network topology.}
\label{fig:networked_plants}
\end{figure}

For the system $\hat{\mathcal H}_p$ in Fig.~\ref{fig:networked_plants}, we have the following lemma:
\begin{lemma}
If the systems $H_{pi}$ are nonlinear NI systems for all $i=1,2,\cdots,N$, then the system $\hat{\mathcal H}_p$ is also a nonlinear NI system.
\end{lemma}
\begin{proof}
According to Definition \ref{def:nonlinear NI}, each nonlinear NI system $H_{pi}$ ($i=1,2,\cdots,N$) must have a corresponding positive definite storage function $V_{pi}(x_{pi})$ such that $\dot V_{pi}(x_{pi})\leq u_{pi}^T\dot y_{pi}$, where $x_{pi}$ is the state of the system $H_{pi}$. We define the storage function for the system $\hat{\mathcal{H}}_p$ as $\hat V_p:=\sum_{i=1}^N V_{pi}(x_{pi})$, which is positive definite. Then
\begin{equation}\label{eq:V_p ineq}
    \dot{\hat V}_p=\sum_{i=1}^N \dot V_{pi}(x_{pi})\leq \sum_{i=1}^N u_{pi}^T\dot y_{pi}=U_p^T\dot Y_p.
\end{equation}
Let $\hat U_p$ and $\hat Y_p$ denote the input and output of the system $\hat{\mathcal{H}}_{p}$, respectively. According to the system setting in Fig.~\ref{fig:networked_plants}, we have
\begin{equation*}
U_p=(\mathcal Q^T\otimes I_m) \hat U_p,\quad \textnormal{and} \quad \hat Y_p=(\mathcal Q\otimes I_m) Y_p.
\end{equation*}
Therefore, we have
\begin{equation}\label{eq:H_p and hat H_p i-o relation}
U_p^TY_p=[(\mathcal Q^T\otimes I_m) \hat U_p)]^TY_p=\hat U_p^T(\mathcal Q\otimes I_m) Y_p=\hat U_p^T\hat Y_p.
\end{equation}
According to (\ref{eq:V_p ineq}) and (\ref{eq:H_p and hat H_p i-o relation}), we obtain the nonlinear NI inequality for the system $\hat {\mathcal H}_p$:
\begin{equation}\label{eq:hat V_p dot}
\dot{\hat V}_p \leq \hat U_p^T\dot{\hat Y}_p.	
\end{equation}
Therefore, $\hat{\mathcal{H}}_p$ is a nonlinear NI system.
\end{proof}

Now we give a definition of output feedback consensus for a network of systems as shown in Fig.~\ref{fig:H_p1_hetero}.
\begin{definition}
A distributed output feedback control law achieves output feedback consensus for a network of systems if $|y_{pi}(t) - y_{pj}(t)|\to 0$ as $t\to +\infty$, $\forall i,j\in{1,2,\cdots,N}$.
\end{definition}

Consider a series of heterogeneous nonlinear OSNI systems $H_{ck}$ $(k=1,2,\cdots,l)$ applied as controllers corresponding to the edges in the network. The OSNI controllers have the following state-space models:
\begin{align}
    \dot x_{ck}(t)=&\ f_{ck}(x_{ck}(t),u_{ck}(t));\label{eq:state ck}\\
    y_{ck}(t)=&\ h_{ck}(x_{ck}(t))\label{eq: output ck},
\end{align}
where $f_{ck}:\mathbb R^q\times \mathbb R^m \to \mathbb R^q$ is a Lipschitz continuous function and $h_{ck}:\mathbb R^q \to \mathbb R^m$ is a class $C^1$ function. They operate independently in parallel and each of them has its own input $u_{ck}\in \mathbb{R}^{m}$ and output $y_{ck}\in\mathbb{R}^{m}$ ($k=1,2,\cdots,l$), which is shown in Fig.~\ref{fig:H_c1_hetero}. We combine the inputs and outputs respectively as the vectors $U_c=\left[u_{c1}^T, u_{c2}^T, \cdots, u_{cl}^T\right]^T\in \mathbb R^{lm\times 1}$ and $Y_c=\left[y_{c1}^T,  y_{c2}^T, \cdots, y_{cl}^T\right]^T\in \mathbb R^{lm\times 1}$.
\begin{figure}[h!]
\centering
\psfrag{H_0}{$\mathcal{H}_c$}
\psfrag{u_p1}{$u_{c1}$}
\psfrag{y_p1}{$y_{c1}$}
\psfrag{u_p2}{$u_{c2}$}
\psfrag{y_p2}{$y_{c2}$}
\psfrag{u_pn}{$u_{cl}$}
\psfrag{y_pn}{$y_{cl}$}
\psfrag{H_p1}{\hspace{-0.05cm}$H_{c1}$}
\psfrag{H_p2}{\hspace{-0.05cm}$H_{c2}$}
\psfrag{H_pn}{\hspace{-0.05cm}$H_{cl}$}
\psfrag{ddd}{\hspace{0.13cm}$\vdots$}
\psfrag{udd}{\hspace{0.15cm}$\vdots$}
\psfrag{odd}{\hspace{0.1cm}$\vdots$}
\hspace{0.5cm}\vspace{-0.3cm}\includegraphics[width=7.5cm]{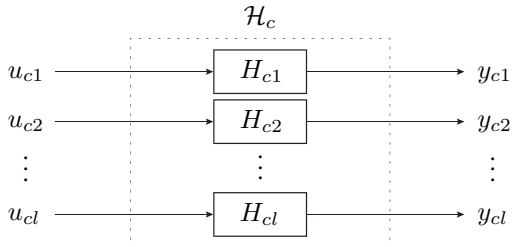}
\caption{System $\mathcal{H}_c$: a nonlinear system consisting of $l$ independent and heterogeneous nonlinear systems $H_{ck}$ $(k=1,2,\cdots,l)$, with independent inputs and outputs combined as the input and output of the networked system $\mathcal{H}_c$.}
\label{fig:H_c1_hetero}
\end{figure}

For the system $\mathcal{H}_c$ in Fig.~\ref{fig:H_c1_hetero}, we have the following lemma:
\begin{lemma}
If the systems $H_{ck}$ are nonlinear OSNI systems for all $k=1,2,\cdots,l$, then the system $\mathcal{H}_{c}$ is a nonlinear OSNI system.
\end{lemma}
\begin{proof}
For every nonlinear OSNI system $H_{ck}$, we have a positive definite storage function $V_{ck}(x_{ck})$ and a constant index $\epsilon_k>0$ such that
\begin{equation}\label{eq:OSNI ineq for controllers}
    \dot V_{ck}(x_{ck})\leq u_{ck}^T\dot y_{ck}-\epsilon_k |\dot y_{ck}|^2,
\end{equation}
where $\epsilon_k$ is the level of output strictness of the system $H_{ck}$. For the system $\mathcal{H}_c$, we define its storage function $V_c$ as $V_c := \sum_{k=1}^l V_{ck}(x_{ck})$, which is positive definite. The time derivative of $V_c$ satisfies the following inequality due to (\ref{eq:OSNI ineq for controllers}):
\begin{align}
    \dot V_c =& \sum_{k=1}^l \dot V_{ck}(x_{ck})
    \leq \sum_{k=1}^l u_{ck}^T\dot y_{ck}-\sum_{k=1}^l\epsilon_k |\dot y_{ck}|^2\notag\\
    \leq & \sum_{k=1}^l u_{ck}^T\dot y_{ck}-\sum_{k=1}^l\epsilon_{min} |\dot y_{ck}|^2\notag\\
         =&\ U_c^T \dot Y_c-\epsilon_{min} |\dot Y_c|^2.\label{eq:hat V_c ineq}
\end{align}
where $\epsilon_{min}=\min\{\epsilon_1,\epsilon_2,\cdots,\epsilon_l\}$. Hence, the system $\mathcal{H}_c$ satisfies the definition of a nonlinear OSNI system. This completes the proof.
\end{proof}

We assume that the following conditions are satisfied.

a) For each individual nonlinear OSNI controller $H_{ck}$ $(k=1,2,\cdots,l)$ with input $u_{ck}(t)$, state $x_{ck}(t)$ and output $y_{ck}(t)$ described by the state-space model (\ref{eq:state ck}), (\ref{eq: output ck}), suppose:

\begin{assumption}\label{assumption:A1}
	Over any time interval $[t_a,t_b]$ where $t_b>t_a$, $y_{ck}(t)$ remains constant if and only if $x_{ck}(t)$ remains constant; i.e., $\dot y_{ck}(t)\equiv 0\iff \dot x_{ck}(t)\equiv 0$. Moreover, $y_{ck}(t)\equiv 0 \iff  x_{ck}(t)\equiv 0$.
\end{assumption}

\begin{assumption}\label{assumption:A2}
	Over any time interval $[t_a,t_b]$ where $t_b>t_a$, $x_{ck}(t)$ remains constant only if $u_{ck}(t)$ remains constant; i.e., $x_{ck}(t)\equiv \bar x_{ck}\implies u_{ck}(t)\equiv \bar u_{ck}$. Moreover, $x_{ck}(t)\equiv 0 \implies u_{ck}(t)\equiv 0$.
\end{assumption}

\begin{figure}[h!]
\centering
\psfrag{H_p}{$\hat {\mathcal{H}}_p$}
\psfrag{U_p}{$\bar {\hat U}_{p}$}
\psfrag{Y_p}{\hspace{-0.1cm}$\hat Y_{p}(t)$}
\psfrag{H_nc}{$\mathcal H_c$}
\psfrag{U_nc}{\hspace{-0.2cm}$U_c(t)$}
\psfrag{Y_nc}{\hspace{0.1cm}$\bar Y_c$}
\includegraphics[width=8cm]{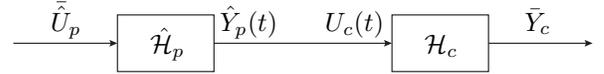}
\caption{Open-loop interconnection of the networked nonlinear NI plants $\hat{\mathcal{H}}_p$ and the networked nonlinear OSNI controllers $\mathcal{H}_c$.}
\label{fig:open_network}
\end{figure}

b) For the open-loop interconnection of the systems $\hat{\mathcal{H}}_p$ and $\mathcal{H}_c$ shown in Fig.~\ref{fig:open_network}, suppose:

\begin{assumption}\label{assumption:A3}
	Given any constant input $\hat U_p(t)\equiv \bar {\hat U}_p$ for the system $\hat{\mathcal H}_p$, we obtain a corresponding output $\hat Y_p(t)$, which is not necessarily constant. Given $\hat Y_p(t)$ as input $U_c(t)$ to the system $\mathcal{H}_c$, if the corresponding output of the system $\mathcal{H}_c$ is a constant $Y_c(t)\equiv \bar Y_c$, then there exists a constant $\gamma\in (0,1)$ such that $\bar {\hat U}_p$ and $\bar Y_c$ satisfy
\begin{equation}\label{eq:assumption 3 ineq}
    \bar {\hat U}_p^T\bar Y_c \leq \gamma \left|\bar {\hat U}_p\right|^2.
\end{equation}
\end{assumption}

Now consider the closed-loop interconnection of the networked plants shown in Fig.~\ref{fig:networked_plants} and the networked controllers shown in Fig.~\ref{fig:H_c1_hetero} in positive feedback, which is depicted in Fig.~\ref{fig:closed_network}. In this paper, the robust output consensus of heterogeneous nonlinear NI plants is achieved by constructing a control system with the block diagram shown in Fig.~\ref{fig:closed_network} and choosing suitable controllers that satisfy certain conditions. The connections between the plants and controllers can be better visualised from the undirected graph, as shown in the example in Fig.~\ref{fig:5_nodes}.

The nodes $p_i$ ($i=1,\cdots,5$ in this example) represent the heterogeneous nonlinear NI plants, while the heterogeneous nonlinear OSNI controllers $c_k$ ($k=1,\cdots,5$ in this example) correspond to the edges. Given any directed graph corresponding to the graph in Fig.~\ref{fig:5_nodes} with the incidence matrix $\mathcal Q$, each edge will have a direction. Then the corresponding connection between the plants and the controller is as shown in Fig.~\ref{fig:Edge_connection}. The controller takes the difference between the outputs of the plants as its input and feeds back its output to the plants with a positive or negative sign corresponding to the edge direction. Each plant takes the sum of all the outputs of the controllers connected to it with correct signs as its input.

\begin{figure}[h!]
\centering
\psfrag{H_p}{$\mathcal{H}_p$}
\psfrag{H_c}{$\mathcal{H}_c$}
\psfrag{H_p1}{$H_{p1}$}
\psfrag{H_p2}{$H_{p2}$}
\psfrag{H_pn}{$H_{pN}$}
\psfrag{H_c1}{$H_{c1}$}
\psfrag{H_c2}{$H_{c2}$}
\psfrag{H_cl}{$H_{cl}$}
\psfrag{U_p}{$U_{p}$}
\psfrag{Y_p}{$Y_{p}$}
\psfrag{U_c}{$U_{c}$}
\psfrag{Y_c}{$Y_{c}$}
\psfrag{ddots}{$\ddots$}
\psfrag{H_np}{$\hat {\mathcal H}_p$}
\psfrag{U_np}{\hspace{0.1cm}$\hat U_p$}
\psfrag{Y_np}{\hspace{0.1cm}$\hat Y_p$}
\psfrag{Q}{\hspace{-0.24cm}$\mathcal{Q}^T\otimes I_m$}
\psfrag{Q_t}{$\mathcal{Q}\otimes I_m$}
\includegraphics[width=8.5cm]{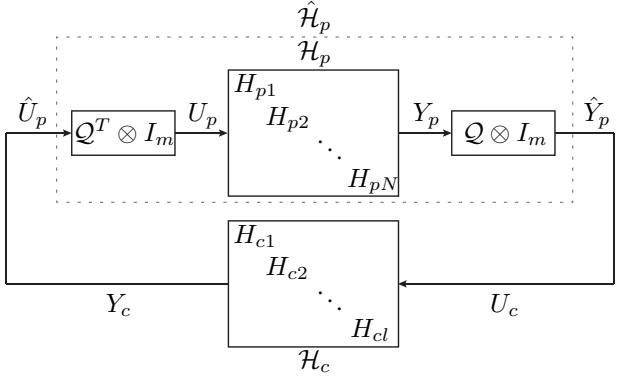}
\caption{Positive feedback interconnection of heterogeneous nonlinear NI plants and nonlinear OSNI controllers according to the directed graph network topology.}
\label{fig:closed_network}
\end{figure}
\begin{figure}[h!]
\centering
\psfrag{p_1}{\large$p_1$}
\psfrag{p_2}{\large$p_2$}
\psfrag{p_3}{\large$p_3$}
\psfrag{p_4}{\large$p_4$}
\psfrag{p_5}{\large$p_5$}
\psfrag{c_1}{\large$c_1$}
\psfrag{c_2}{\large$c_2$}
\psfrag{c_3}{\large$c_3$}
\psfrag{c_4}{\large$c_4$}
\psfrag{c_5}{\large$c_5$}
\includegraphics[width=6cm]{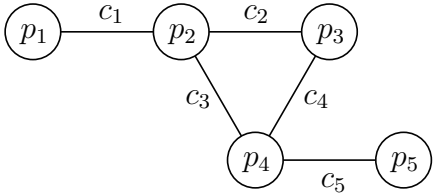}
\caption{An example of the networked connection of plants and controllers.}
\label{fig:5_nodes}
\end{figure}

\begin{figure}[h!]
\centering
\psfrag{p_1}{\large$p_1$}
\psfrag{p_2}{\large$p_2$}
\psfrag{c_1}{\large$c_1$}
\psfrag{+}{\scriptsize$+$}
\psfrag{-}{\scriptsize$-$}
\psfrag{d_1}{$d_1=0$}
\psfrag{d_2}{$d_2=0$}
\psfrag{y_1}{$y_{p1}$}
\psfrag{y_2}{$y_{p2}$}
\psfrag{y_c}{\hspace{-0.15cm}$y_{c1}=H_{c1}(y_{p1}-y_{p2})$}
\vspace{-0.3cm}\includegraphics[width=7.5cm]{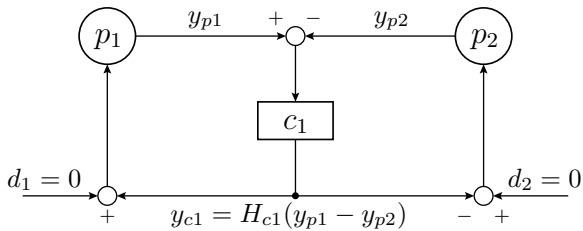}
\caption{Detailed block diagram corresponding to a pair of nodes connected by an edge.}
\label{fig:Edge_connection}
\end{figure}

\begin{theorem}\label{theorem:consensus}
Consider an undirected connected graph $\mathcal G$ that models the communication links for a network of heterogeneous nonlinear NI systems $H_{pi}$ $(i=1,2,\cdots,N)$ as shown in Fig.~\ref{fig:H_p1_hetero}, and any directed graph corresponding to $\mathcal G$ with the incidence matrix $\mathcal Q$. Also, consider the heterogeneous nonlinear OSNI control laws $H_{ck}$ $(k=1,2,\cdots,l)$ for all the edges. Suppose Assumptions \ref{assumption:A1}, \ref{assumption:A2} and \ref{assumption:A3} are satisfied and the storage function, defined as
\begin{equation*}
    W:=\hat V_p+V_c-\hat Y_p^TY_c\ ,
\end{equation*}
is positive definite, where $\hat V_p$ and $V_c$ are positive definite storage functions that satisfy (\ref{eq:V_p ineq}) for the system $\hat{\mathcal{H}}_p$ and (\ref{eq:hat V_c ineq}) for the system $\mathcal{H}_c$, respectively. Here, $\hat Y_p$ and $Y_c$ are outputs of the systems $\hat {\mathcal{H}}_p$ and $\mathcal{H}_c$, respectively. Then the robust output feedback consensus can be achieved via the protocol
\begin{equation*}
    U_p = (\mathcal{Q}^T\otimes I_m)\mathcal{H}_c\left((\mathcal{Q}\otimes I_m)Y_p\right),
\end{equation*}
or equivalently,
\begin{equation*}
    u_{pi}=\sum_{k=1}^l q_{ki}H_{ck}\left(\sum_{j=1}^N q_{kj}y_{pj}\right),
\end{equation*}
for each plant $pi$, as shown in Fig.~\ref{fig:closed_network}, where $\sum_{j=1}^N q_{kj}y_{pj}$ represents the difference between the outputs of the plants connected by the edge $e_k$.
\end{theorem}
\begin{proof}
We apply the Lyapunov's direct method and take the time derivative of the storage function $W$. According to (\ref{eq:V_p ineq}) and (\ref{eq:hat V_c ineq}), we have
\begin{align}
    \dot W=&\ \dot {\hat V}_p+\dot V_c-\dot{\hat Y}_p^TY_c - \hat Y_p^T\dot{Y}_c\notag\\
    =&\ \dot{\hat V}_p + \dot V_c-\hat U_p^T\dot {\hat Y}_p- U_c^T \dot Y_c\notag\\
    \leq &-\epsilon_{min} \left|\dot Y_c\right|^2\leq \ 0.\label{eq:W dot}
\end{align}
Hence, the closed-loop system is at least Lyapunov stable. Now we apply LaSalle's invariance principle. According to (\ref{eq:W dot}), $\dot{W}$ can remain zero only if $\epsilon_{min} |\dot Y_{c}|^2$ remains zero, which means $\dot y_{ck}(t)$ remain zero for all $k=1,2,\cdots,l$. According to Assumptions \ref{assumption:A1} and \ref{assumption:A2}, for the system $H_{ck}$, $\dot y_{ck}(t)\equiv 0$ implies $\dot x_{ck}(t)\equiv 0$, which holds only if $u_{ck}(t)\equiv \bar u_{ck}$. In other words, for all $k=1,2,\cdots,l$, the controllers $H_{ck}$ are in steady-state; i.e., $u_{ck}(t)\equiv \bar u_{ck}$, $x_{ck}(t)\equiv \bar x_{ck}$ and $y_{ck}(t)\equiv \bar y_{ck}$. Consider the setting in Fig.~\ref{fig:closed_network}, in which $U_c(t)$ and $Y_c(t)$ are all constant vectors; i.e., $U_c(t)\equiv \bar U_c$ and $Y_c(t)\equiv \bar Y_c$. According to the closed-loop setting that $\hat U_p(t)\equiv Y_c(t)$, $\hat U_{p}(t)\equiv \bar {\hat U}_p$ is also a constant vector. The inequality (\ref{eq:assumption 3 ineq}) implies
\begin{equation*}
    \bar {\hat U}_p^T\bar{ Y}_c=\left|\bar {\hat U}_p\right|^2\leq \gamma \left|\bar {\hat U}_p\right|^2.
\end{equation*}
This can only hold if $\bar {\hat U}_p=0$, which implies $\bar{Y}_c=0$. Hence, $\bar y_{ck}=0$ for all $k=1,2,\cdots,l$. According to Assumptions \ref{assumption:A1} and \ref{assumption:A2}, $\bar y_{ck}=0$ implies $\bar x_{ck}=0$ and then $\bar u_{ck}=0$ for all $k=1,2,\cdots,l$. This implies $y_{pi}(t)\equiv y_{pj}(t)$ for all $(v_i,v_j)\in \mathcal{E}$, which means output consensus is achieved for all of the heterogeneous nonlinear NI plants. Otherwise, $\dot W$ cannot remain at zero and $W$ will keep decreasing until output consensus is achieved or $W=0$, which also implies output consensus. This completes the proof.
\end{proof}
\begin{remark}
The protocol in Theorem \ref{theorem:consensus} is robust with respect to small parameter perturbation for heterogeneous nonlinear NI plants connected in a network. Indeed, for physical systems, NI property is invariant to parameter variations so that consensus can be always guaranteed as long as the conditions required in Theorem \ref{theorem:consensus} are satisfied.
\end{remark}

We now provide some typical first-order and second-order dynamical systems as possible choices for nonlinear OSNI controllers.
\begin{lemma}\label{lemma:1st order OSNI system}
Consider a first-order system with the state-space model:
	\begin{align*}
		\dot x(t) =&\ \rho(x(t))+\alpha u(t);\\
		y(t) =&\ x(t)
	\end{align*}
where $x(t)$, $u(t)$ and $y(t)$ are scalar functions of time, $\rho:\mathbb R\to \mathbb R$ is a Lipschitz continuous function and $\alpha > 0$ is a constant. If the function $V$, given by $V(x)=-\frac{1}{\alpha}\int_0^x \rho(z)dz$ is positive definite, then the system is nonlinear OSNI with level of strictness $\epsilon\in(0,\frac{1}{\alpha}]$ and with $V$ being a storage function.
\end{lemma}
\begin{proof}
	Let us define $D(x)=\dot V(x)-\left(u\dot y-\epsilon \dot y^2\right)$. We prove in the following that $D(x)\leq 0$ for $\epsilon\in (0,\frac{1}{\alpha}]$.
	\begin{equation*}
	\begin{aligned}
		D(x)=&\ \dot V(x)-\left(u\dot y-\epsilon \dot y^2\right)\\
		=&\ \frac{\partial V(x)}{\partial x}\dot x-u\dot x+\epsilon \dot x^2\\
		=&\ \dot x \left[\frac{\partial V(x)}{\partial x}-u+\epsilon \dot x\right]\\
		=&\left(\rho(x)+\alpha u\right)\left[-\frac{1}{\alpha}\rho(x)-u+\epsilon\left(\rho(x)+\alpha u\right)\right]\\
		=&\left(\epsilon-\frac{1}{\alpha}\right)\left(\rho(x)+\alpha u\right)^2\leq \ 0
	\end{aligned}
	\end{equation*}
when $\epsilon\in \left(0,\frac{1}{\alpha}\right]$. Therefore, this system is a nonlinear OSNI system according to Definition \ref{def:nonlinear OSNI}.
\end{proof}
\begin{lemma}\label{lemma:2nd order OSNI system}
Consider a second order system with the following state-space model:
	\begin{align*}
		\left[\begin{matrix}\dot x_1(t)\\ \dot x_2(t)\end{matrix}\right] =& \left[\begin{matrix}x_2(t)\\ \eta(x_1(t))-\beta x_2(t)+\alpha u(t)\end{matrix}\right];\\
		y(t) =&\ x_1(t)
	\end{align*}
where $x_1(t)$, $x_2(t)$, $u(t)$ and $y(t)$ are scalar functions of time, $\eta:\mathbb R\to \mathbb R$ is a Lipschitz continuous function and $\alpha>0$ and $\beta>0$ are constants. If the function $V$, given by $V(x_1,x_2)=-\frac{1}{\alpha}\int_0^{x_1} \eta(z)dz+\frac{1}{2\alpha}x_2^2$ is positive definite, then the system is nonlinear OSNI with level of strictness $\epsilon\in(0,\frac{\beta}{\alpha}]$ and with $V$ being a storage function.  
\end{lemma}
\begin{proof}
	Define $D(x_1,x_2)=\dot V(x_1,x_2)-\left(u\dot y-\epsilon \dot y^2\right)$. We prove in the following that $D(x_1,x_2)\leq 0$ for $\epsilon\in (0,\frac{\beta}{\alpha}]$.
	\begin{equation*}
	\begin{aligned}
		D(x_1,x_2)=&\ \dot V(x_1,x_2)-\left(u\dot y-\epsilon \dot y^2\right)\\
		=&\ \frac{\partial V(x_1,x_2)}{\partial x_1}\dot x_1+\frac{\partial V(x_1,x_2)}{\partial x_2}\dot x_2-u\dot x_1+\epsilon \dot x_1^2\\
		=& -\frac{1}{\alpha}\eta (x_1)x_2+\frac{1}{\alpha}x_2\left[\eta(x_1)-\beta x_2+\alpha u\right]\\
		&-ux_2+\epsilon x_2^2\\
		=&\ \left(\epsilon-\frac{\beta}{\alpha}\right)x_2^2\leq \ 0
		\end{aligned}
	\end{equation*}
when $\epsilon\in \left(0,\frac{\beta}{\alpha}\right]$. Hence, this system is a nonlinear OSNI system according to Definition \ref{def:nonlinear OSNI}.
\end{proof}

\section{Example}
This section illustrates the robust output feedback consensus protocol described in Theorem \ref{theorem:consensus} with an example of networked heterogeneous pendulum systems.

Consider three pendulum systems connected by an undirected connected graph $\mathcal G$ as shown in Fig.~\ref{fig:3_nodes}.
\begin{figure}[h!]
\centering
\psfrag{1}{\hspace{-0.02cm}\large$1$}
\psfrag{2}{\hspace{0.0cm}\large$2$}
\psfrag{3}{\hspace{-0.005cm}\large$3$}
\psfrag{e_1}{$e_1$}
\psfrag{e_2}{$e_2$}
\includegraphics[width=4cm]{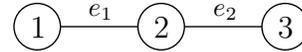}
\caption{An undirected and connected graph consisting of three nodes.}
\label{fig:3_nodes}
\end{figure}
The Laplacian matrix of graph $\mathcal G$ is $\mathcal L_3=\left[\small \begin{matrix}1&-1&0\\-1&2&-1\\0&-1&1\end{matrix}
\right]$. To obtain a directed graph corresponding to $\mathcal G$, we can arbitrarily decide the direction of each edge. If we decide the directions of the edges as $e_1=(v_1,v_2)$ and $e_2=(v_2,v_3)$, then the incidence matrix of the directed graph corresponding to $\mathcal G$ is $\mathcal Q = \left[\small\begin{matrix}1&-1&0\\0&1&-1 \end{matrix}\right]$.

These pendulum systems have the following state-space model:
\begin{equation*}
\begin{aligned}
    \left[\begin{matrix}\dot x_1\\\dot x_2\end{matrix}\right]=&\left[\begin{matrix} x_2\\ \frac{1}{ml^2}\left(-\kappa x_1-mgl\sin {x_1}+u_1\right)\end{matrix}\right];\\
    y_1=&\ x_1
\end{aligned}
\end{equation*}
where $m$ is the mass of each bob, $l$ is the length of each rod, $\kappa$ is the spring constant of a torsional spring installed in each pivot and $g\approx 9.8m/s^2$ is the gravitational acceleration. $m$, $l$ and $\kappa$ are different for the three heterogeneous networked pendulums. For each pendulum system, the system state $x_1$ is the counterclockwise angular displacement from the vertically downward position and $x_2$ is the system angular velocity. The system input $u$ is an external torsional force in the counterclockwise direction, and $y$ is the system output. The system is a nonlinear NI system with the storage function $V_1(x_1,x_2)=\frac{1}{2}\kappa x_1^2+\frac{1}{2}ml^2x_2^2+mgl(1-\cos{x_1})$.

According to Lemma \ref{lemma:1st order OSNI system}, we choose the following nonlinear OSNI system as the control law corresponding to each edge:
\begin{equation*}
\begin{aligned}
    \dot x_c=&-\beta x_c-\phi x_c^3+\alpha u_c;\\
    y_c=&\ x_c
\end{aligned}
\end{equation*}
where $\beta>0$, $\phi>0$ and $\alpha>0$ are constants. The nonlinear OSNI property of this system can be proved with the storage function $V_c(x_c)=\frac{\beta}{2\alpha}x_c^2+\frac{\phi}{4\alpha}x_c^4$.

Suppose the pendulums have the following parameters:
\begin{align*}
	\textnormal{pendulum 1:}&\ m_1=1kg,\hspace{2.75mm} l_1=0.5m\ \textnormal{and}\ \kappa_1=3Nm/rad;\\
	\textnormal{pendulum 2:}&\ m_2=1.5kg, l_2=0.3m\ \textnormal{and}\ \kappa_2=5Nm/rad;\\
	\textnormal{pendulum 3:}&\ m_3=0.5kg, l_3=0.8m\ \textnormal{and}\ \kappa_3=6Nm/rad.
\end{align*}

The parameters for the controllers are chosen to be:
\begin{align*}
	\textnormal{controller 1:}&\ \beta_1=10,\ \phi_1=15\ \textnormal{and}\ \alpha_1=20;\\
	\textnormal{controller 2:}&\ \beta_2=20,\ \phi_2=5\ \hspace{1.7mm}\textnormal{and}\ \alpha_2=30.
\end{align*}

We choose different parameters for the plants and the controllers to demonstrate that our protocal allows both the plants and the controllers to be heterogeneous systems. According to Theorem~\ref{theorem:consensus}, we control the pendulums with the distributed control law: $u_{p1}=H_{c1}(y_{p1}-y_{p2})$, $u_{p2}=-H_{c1}(y_{p1}-y_{p2})+H_{c2}(y_{p2}-y_{p3})$ and $u_{p3}=-H_{c2}(y_{p2}-y_{p3})$, respectively. Here, $H_{ck}(\cdot)$ represents the output of controller $c_k$. It can be verified that Assumptions \ref{assumption:A1}, \ref{assumption:A2} and \ref{assumption:A3} are satisfied and the storage function of the entire networked system is positive definite. As shown in Fig.~\ref{fig:example_figure}, the pendulum systems approach the same limit trajectory under the effect of the heterogeneous nonlinear OSNI controllers.

\begin{figure}[h!]
\centering
\psfrag{Output}{Output}
\psfrag{time (s)}{Time (s)}
\psfrag{Output Feedback consensus of pendulums}{\hspace{0.1cm}Output Feedback Consensus of Pendulums}
\psfrag{Pendulum sys 1}{\small Pendulum 1}
\psfrag{Pendulum sys 2}{\small Pendulum 2}
\psfrag{Pendulum sys 3}{\small Pendulum 3}
\psfrag{-1}{-1}
\psfrag{-0.5}{-0.5}
\psfrag{0}{0}
\psfrag{0.5}{0.5}
\psfrag{1}{1}
\psfrag{5}{5}
\psfrag{10}{10}
\psfrag{15}{15}
\psfrag{20}{20}
\psfrag{25}{25}
\psfrag{30}{30}
\includegraphics[width=9cm]{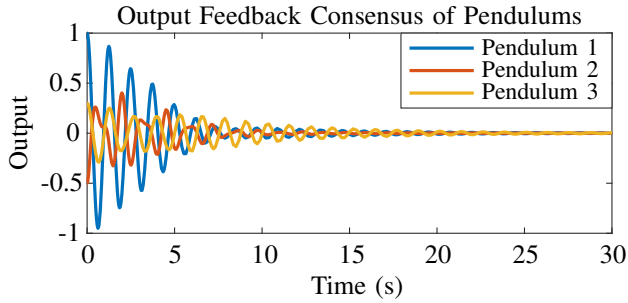}
\caption{Robust output feedback consensus for a network of heterogeneous pendulum systems with heterogeneous nonlinear OSNI controllers applied.}
\label{fig:example_figure}
\end{figure}

\section{Conclusion}
This paper provides a protocol for the output feedback consensus problem of heterogeneous nonlinear NI systems. For a network of heterogeneous nonlinear NI systems connected by an undirected and connected graph, heterogeneous edge-based nonlinear OSNI controllers can be applied in positive feedback through a network topology leading to convergence of the outputs of the nonlinear NI plants to a common limit trajectory if certain conditions are satisfied. This protocol is robust with respect to parameter perturbation in the system models of the nonlinear NI plants and the nonlinear OSNI controllers so that any network of heterogeneous nonlinear NI systems can be synchronised as long as their nonlinear NI property is preserved and certain conditions are satisfied. Some typical first-order and second-order nonlinear systems are also provided as possible choices for nonlinear OSNI controllers.

\bibliographystyle{IEEEtran}

\end{document}